\documentclass[submission,copyright,creativecommons]{eptcs}
\usepackage{graphicx}
\usepackage{amssymb}
\usepackage{amsthm}
\usepackage{amsmath}
\usepackage{hyperref}
\providecommand{\event}{GandALF 2013} 
\usepackage{breakurl}             

\title{Weighted Automata and
\\Monadic Second Order Logic}
\author{Nadia Labai 
\institute{Faculty of Computer Science}
\institute{Technion--Israel Institute of Technology}
\email{nadia@cs.technion.ac.il}
\and
Johann A. Makowsky 
\institute{Faculty of Computer Science}
\institute{Technion--Israel Institute of Technology}
\email{\quad janos@cs.technion.ac.il}
}
\def\titlerunning{Weighted Automata}
\def\authorrunning{N. Labai \& J.A. Makowsky}
\begin{document}
\maketitle
\newcommand{\FF}{{\rm FF}}
\newcommand{\cH}{\mathcal{H}}
\newif\ifskip
\skiptrue
\theoremstyle{plain}
\newtheorem{theorem}{Theorem}
\newtheorem{definition}[theorem]{Definition}
\newtheorem{lemma}[theorem]{Lemma}
\newtheorem{proposition}[theorem]{Proposition}
\newtheorem{corollary}[theorem]{Corollary}
\newtheorem{remark}[theorem]{Remark}
\newenvironment{Proof}{{\bf Proof:}\ }{$\Box$\newline}
\newtheorem{problem}{Problem}
\newtheorem{exs}[theorem]{Examples}
\newtheorem{conclusion}[theorem]{Conclusion}
\newtheorem{rem}[theorem]{Remark}
\newcounter{ourown}

\ifskip
\else
\theoremstyle{definition}
\theoremstyle{remark}
\newtheorem{rem}[theorem]{Remark}
\newtheorem{example}[theorem]{Example}
\newtheorem{examples}[theorem]{Examples}
\newtheorem{remarks}[theorem]{Remarks}
\fi 
\newenvironment{renumerate}{\begin{enumerate}}{\end{enumerate}}
\renewcommand{\theenumi}{\roman{enumi}}
\renewcommand{\labelenumi}{(\roman{enumi})}
\renewcommand{\labelenumii}{(\roman{enumi}.\alph{enumii})}

\newcommand{\FORM}{\mathcal{FF}}
\newcommand{\LL}{\langle\langle}
\newcommand{\RR}{\rangle\rangle}
\newcommand{\cL}{\mathcal{L}}
\newcommand{\cF}{{\mathcal{F}}}
\newcommand{\cT}{{\mathcal{T}}}
\newcommand{\cS}{{\mathcal{S}}}
\newcommand{\cM}{{\mathcal{M}}}
\newcommand{\cR}{{\mathcal{R}}}
\newcommand{\cA}{{\mathcal{A}}}
\newcommand{\card}[3]{card_{#1,\bar{#2}}(#3(\bar{#2}))}
\newcommand{\cardm}[3]{card_{#1,#2}(#3(#2))}

\newcommand{\WFF}{\mathbf{WFF}}
\newcommand{\MSOL}{\mathbf{MSOL}}
\newcommand{\SOL}{\mathbf{SOL}}
\newcommand{\CMSOL}{\mathbf{CMSOL}}
\newcommand{\SEN}{\mbox{\bf SEN}}
\newcommand{\WFTF}{\mbox{\bf WFTF}}
\newcommand{\FOL}{\mbox{\bf FOL}}
\newcommand{\TFOF}{\mbox{\bf TFOF}}
\newcommand{\TFOL}{\mbox{\bf TFOL}}
\newcommand{\FOF}{\mbox{\bf FOF}}
\newcommand{\NNF}{\mbox{\bf NNF}}
\newcommand{\N}{{\mathbb N}}
\newcommand{\bN}{{\mathbb N}}
\newcommand{\bR}{{\mathbb R}}
\newcommand{\HF}{\mbox{\bf HF}}
\newcommand{\CNF}{\mbox{\bf CNF}}
\newcommand{\PNF}{\mbox{\bf PNF}}
\newcommand{\QF}{\mbox{\bf QF}}
\newcommand{\DNF}{\mbox{\bf DNF}}
\newcommand{\DISJ}{\mbox{\bf DISJ}}
\newcommand{\CONJ}{\mbox{\bf CONJ}}
\newcommand{\Ass}{\mbox{Ass}}
\newcommand{\Var}{\mbox{Var}}
\newcommand{\Support}{\mbox{Support}}
\newcommand{\V}{\mbox{\bf Var}}
\newcommand{\fA}{{\mathfrak A}}
\newcommand{\fB}{{\mathfrak B}}
\newcommand{\fN}{{\mathfrak N}}
\newcommand{\fZ}{{\mathfrak Z}}
\newcommand{\fQ}{{\mathfrak Q}}
\newcommand{\Aa}{{\mathfrak A}}
\newcommand{\Bb}{{\mathfrak B}}
\newcommand{\Cc}{{\mathfrak C}}
\newcommand{\Gg}{{\mathfrak G}}
\newcommand{\Ww}{{\mathfrak W}}
\newcommand{\Rr}{{\mathfrak R}}
\newcommand{\Nn}{{\mathfrak N}}
\newcommand{\Zz}{{\mathfrak Z}}
\newcommand{\Qq}{{\mathfrak Q}}
\newcommand{\F}{{\mathbf F}}
\newcommand{\T}{{\mathbf T}}
\newcommand{\Z}{{\mathbb Z}}
\newcommand{\R}{{\mathbb R}}
\newcommand{\C}{{\mathbb C}}
\newcommand{\Q}{{\mathbb Q}}
\newcommand{\MT}{\mbox{MT}}
\newcommand{\TT}{\mbox{TT}}
\newcommand{\RMSOL}{\mathbf{RMSOL}}
\newcommand{\WMSOL}{\mathbf{WMSOL}}
\newcommand{\bMSOL}{\mathbf{bMSOL}}
\newcommand{\MSOLEVAL}{\mathbf{MSOLEVAL}}
\newcommand{\SOLEVAL}{\mathbf{SOLEVAL}}
\newcommand{\tw}{\mathrm{tv}}
\newcommand{\TRUE}{\mathrm{TRUE}}
\newcommand{\FALSE}{\mathrm{FALSE}}

\newcommand{\mc}[1]{\mathcal{#1}}

\newif\ifappendix
\appendixfalse

\begin{abstract}
Let $\cS$ be a commutative semiring.
M. Droste and P. Gastin 
have introduced in 2005 weighted  monadic
second order logic $\WMSOL$ with weights in $\cS$.
They use a syntactic fragment $\RMSOL$ of $\WMSOL$
to characterize word functions (power series) recognizable 
by weighted automata, where the semantics of quantifiers is 
used both as arithmetical
operations and, in the boolean case, as quantification.

Already in 2001, B. Courcelle, J.Makowsky and U. Rotics have introduced a
formalism for graph parameters definable in Monadic Second order
Logic, here called $\MSOLEVAL$ with values in a ring $\cR$. 
Their framework can be easily adapted to semirings $\cS$.
This formalism clearly separates the logical part from the
arithmetical part and also applies to word functions.

In this paper we give two proofs that $\RMSOL$ and $\MSOLEVAL$ with values in $\cS$ 
have the same expressive power over words. 
One proof shows directly that $\MSOLEVAL$ captures the functions recognizable by weighted automata.
The other proof shows how to translate the formalisms from one into the other.

\ifskip
\else
In this paper we show, using model theoretic methods, and
Hankel matrices (in the case of
fields) or stable semimodules (in the case of semirings)
that the word functions in $\MSOLEVAL$ 
with values in $\cS$ are exactly the functions
recognizable by weighted automata.
Our formalism can also be used to study recognizable
families of graphs and fits the framework of Courcelle's
general study of monadic second order logic of graphs,
as exemplified in the recent monumental monograph by B.Courcelle and J. Engelfriet.

We also compare the two formalisms directly and show that they
are inter-translatable, hence the characterization of M. Droste and P. Gastin
and ours can be deduced from each other.
However, our model theoretic approach
sheds a different light on recognizability by using model theoretic tools
with a wide range of further applications.
\fi

\end{abstract}
\section{Introduction}

Let $f$ be a function from relational structures of a fixed relational vocabulary $\tau$
into some field, ring, or a commutative semiring $\cS$ which is invariant under $\tau$-isomorphisms. 
$\cS$ is called a {\em weight structure}.
In the case where the structures are graphs, such a function is called a graph parameter,
or, if $\cS$ is a polynomial ring, a graph polynomial.
In the case where the structures are words, it is called a word function. 

The study of definability of graph parameters and graph polynomials
in Monadic Second Order Logic $\MSOL$ was initiated in
\cite{ar:CourcelleMakowskyRoticsDAM} 
and further developed in
\cite{ar:MakowskyTARSKI,ar:KotekMakowskyZilber08}.
For a weight structure $\cS$
we denote the  set of functions of $\tau$-structures definable in $\MSOL$
by $\MSOLEVAL(\tau)_{\cS}$, or if the context is clear, just
by $\MSOLEVAL_{\cS}$. 
The original purpose for studying functions in
$\MSOLEVAL_{\cS}$ was to prove an analogue to Courcelle's celebrated theorem
for polynomial rings as weight structures,
which states that graph parameters $f \in \MSOLEVAL_{\cS}$ are computable in linear
time for graphs of fixed tree-width, \cite{ar:CourcelleMakowskyRoticsDAM}, 
and various generalizations thereof.
$\MSOLEVAL$ can be seen as an analogue of the {\em Skolem elementary functions} aka {\em lower elementary functions}, \cite{ar:Skolem1962,ar:Volkov2010},
adapted to the framework of {\em meta-finite model theory} as defined in \cite{ar:GraedelGurevich}.

\ifskip
\else
Graph parameters in $\MSOLEVAL_{\cS}$  have many other interesting properties,
cf. \cite{phd:Kotek}. In particular, if we deal with sum-like binary operations of labeled
graphs, the corresponding Hankel matrices (aka connection matrices) have finite rank, 
if $\cF$ is a field, \cite{ar:GodlinKotekMakowsky08,pr:KotekMakowsky2012}.
When the context requires $\cS$ to be a field, rather than to be a commutative semiring, we write $\cF$ instead of $\cS$.
\fi 

In \cite{ar:DrosteGastin2007} a different formalism to define $\cS$-valued word  functions was introduced, which
the authors called {\em weighted monadic second order logic $WMSOL$}, and used a fragment, $\RMSOL$, of it to prove 
that a word function is recognized by a weighted automaton iff it is definable in $\RMSOL$.
This can be seen as an analogue of the
B\"uchi-Elgot-Trakhtenbrot Theorem characterizing regular languages for the case of
weighted (aka multiplicity) automata.

\subsection*{Main results}

Our main results explore various features of the two formalisms $\MSOLEVAL$ and $\RMSOL$ for word functions
with values in a semiring $\cS$. In the study of $\MSOLEVAL$ we show how {\em model theoretic tools} can be used
to characterize the word functions  in $\MSOLEVAL$ as the fuctions recognizable by weigthed automata.
This complements the automata theoretic approach used in the study of weighted automata, \cite{bk:HandbookWeightedAutomata,ar:DrosteKuich2013}.
In particular, we give two proofs that $\RMSOL$ and $\MSOLEVAL$ with values in a semiring $\cS$ 
have the same expressive power over words. 
To see this we show the following for a word function $f$ with values in $\cS$:
\begin{renumerate}
\item
If $f$ is definable in $\MSOLEVAL$, it is contained in a finitely generated stable
semimodule of word functions, Theorem \ref{th:stable}.
\item
If $f$ is recognizable by some weighted automaton, it is definable in $\MSOLEVAL$, the ``if'' direction of 
Theorem \ref{th:main-1}.
\item
If $f$ is definable in $\RMSOL$, we can translate it, using Lemma \ref{le:bMSOL},
into an expression in $\MSOLEVAL$,
Theorem \ref{th:transl}.
\item
If $f$ is definable in $\MSOLEVAL$, we can, again using Lemma \ref{le:bMSOL},
translate it into an expression in $\RMSOL$,
Theorem \ref{th:transl-1}.
\end{renumerate}

Items (i) and (ii) together with a classical characterization of recognizable word functions
in terms of finitely generated stable semimodules, Theorem \ref{Jacob},  cf.
\cite{bk:BerstelReutenauer,phd:Jacob,ar:Fliess1974},
give us a direct proof that
$\MSOLEVAL$ captures the functions recognizable by weighted automata.
To prove item (i) we rely on and extend results about $\MSOLEVAL$  from
\cite{ar:MakowskyTARSKI,ar:GodlinKotekMakowsky08,phd:Kotek}.

Items (iii) and (iv) together 
show how to translate the formalisms $\RMSOL$ and $\MSOLEVAL$ into each other.
Lemma \ref{le:bMSOL}
also shows how the fragment $\RMSOL$ of the weighted logic $\WMSOL$ comes into play.

\ifskip
The point of separating (i) and (ii) from (ii) and (iv) and giving {\em two} proofs of
Theorem \ref{th:main-1} is to show that the model theoretic methods developed in the 1950ties
and further developed in \cite{ar:MakowskyTARSKI} suffice to characterize the functions recognized
by weighted automata.
\else
\fi 

\subsection*{Background and outline of the paper}
We assume the reader is familiar with Monadic Second Order Logic and Automata Theory as described in
\cite{bk:EF95,bk:BerstelReutenauer} or similar references.
In Section \ref{se:msoleval} we introduce $\MSOLEVAL$ by example, which suffices for our purposes.
\ifappendix
A full definition is given in Appendix \ref{app:inductive}.
\else
A full definition is given in Appendix \ref{appn:inductive}.
\fi 
In Section \ref{se:aut} we show that the word functions which are recognizable
by a weighted automaton are exactly the word functions definable in $\MSOLEVAL$. 
In Section \ref{se:comp} we give the exact definitions of $\WMSOL$ and $\RMSOL$, and present translations between
$\MSOLEVAL$ and $\RMSOL$ in both directions.
In Section \ref{se:conclu} we draw our conclusions.

\section{Definable word functions}
\label{se:msoleval}
Let $\cS$ be a  commutative semiring. 
We denote  structures over a finite relational signature (aka vocabulary) $\tau$
by $\cA$ and their underlying universe by $A$.
The class  of functions in
$\MSOLEVAL_{\cS}$ 
consists of the functions which map relational structures into $\cS$, and which are
definable in 
Monadic Second Order Logic $\MSOL$.
The functions in
$\MSOLEVAL_{\cS}$  are represented as terms associating with each $\tau$-structure
$\mathcal{A}$ a polynomial $p(\mathcal{A}, \bar{X}) \in \cS[\bar{X}]$.
The class of such polynomials is defined inductively where monomials are products of
constants in $\cS$ and indeterminates in $\bar{X}$ and the product ranges over
elements $a$ of $A$ which satisfy an $\MSOL$-formula $\phi(a)$.
The polynomials are then defined as sums of monomials where the sum ranges over {\em unary}
relations $U \subseteq A$ satisfying an $\MSOL$-formula $\psi(U)$.
The word functions are obtained by substituting elements of $\cS$ for the indeterminates.
The details of the definition of $\MSOLEVAL_{\cS}$ are given 
\ifappendix
in Appendix \ref{app:inductive}, which is taken almost
literally from \cite{ar:KotekMakowskyZilber11} and also treats the more general case of full second order logic $\SOLEVAL$.
Here we explain 
\else
at the end of this section. We first explain
\fi 
the idea of $\MSOLEVAL_{\cS}$ by examples for the case where structures represent words
over a fixed alphabet $\Sigma$.

\ifappendix
\else
\subsection{Guiding examples}
\fi 
Let $f: \Sigma^{\star} \rightarrow \cS$
be an $\cS$-valued function on words over the alphabet $\Sigma$ and let $w$ be a word in $\Sigma^{\star}$.
We call such functions {\em word functions}, following \cite{ar:CarlylePaz1971,ar:Cobham1978}. 
They are also called {\em formal power series} in \cite{bk:BerstelReutenauer},
where the indeterminates are indexed by words and the coefficient of $X_w$ is $f(w)$.

We denote by $w[i]$ the letter at position $i$ in $w$, and 
by $w[U]$
the word induced by $U$,
for $U$ a set of positions in $w$. 
We denote the length of a word $w$ by $\ell(w)$ and the concatenation of two words $u,v \in \Sigma^{\star}$ by $u \circ v$.
We denote by $[n]$ the set $\{1,2, \ldots , n\}$.

We will freely pass between words and structures representing words.
For the sequel, let $\Sigma =\{0,1\}$ and $w \in \{0,1\}^{\star}$ be represented by the structure
$$
\mathcal{A}_w = \langle \{0\} \cup  [\ell(w)], <^w, P_0^w, P_1^w \rangle.
$$
$P_0^w, P_1^w \subseteq [\ell(w)]$ and $P_0^w \cap P_1^w = \emptyset$ and $P_0^w \cup P_1^w = [\ell]$.

As structures are always non-empty, the universe of a word $w$ is represented by
a structure containing the zero position $[n] \cup \{0\}=\{0,1, \ldots, n\}$.
So strictly speaking the size of the structure of the empty word is one,
and of a word of length $n$ it is $n+1$.
The zero position, represented by $0$, has no letter attached to it, and
the elements of the structure different from $0$ represent positions in the word which carry letters. 
The positions in $P_0^w$ carry the letter $0$ and
the positions in $P_1^w$ carry the letter $1$.

\begin{exs}
\label{ex:1}
In the following examples the functions are word functions
with values in the ring $\Z$ or the polynomial ring $\Z[X]$.
\begin{renumerate}
\item 
The function $\sharp_1(w)$ counts the number of occurrences of $1$ in a word $w$
and can be written as 
$$\sharp_1(w) = \sum_{i \in [n]: P_1(i)} 1.$$
\item 
The polynomial $X^{\sharp_1(w)}$ 
can be written as 
$$X^{\sharp_1(w)} = \prod_{i \in [n]: P_1(i)} X.$$
\item 
Let $L$ be a regular language defined by the $\MSOL$-formula $\phi_L$.
The generating function of the number of (contiguous) occurrences of words $u \in L$ in a word $w$,
can be written as 
$$\sharp_L(w) = \sum_{U \subseteq [n]: w[U] \models \psi_L} \prod_{i \in U} X,$$
where $\psi_L(U)$ says that $U$ is an interval and  $\phi_L^U$, the relativization of $\phi_L$ to $U$,
holds.
\item 
The functions 
$\mathrm{sq}(w) = 2^{\ell(w)^2}$ and
$\mathrm{dexp}(w) = 2^{2^{\ell(w)}}$
are not representable in $\MSOLEVAL_{\mathcal{F}}$.
\end{renumerate}
\end{exs}

The {\em tropical semiring} $\cT_{min}$ is the semiring with universe 
$\R \cup \{\infty\}$,
consisting of the real numbers augmented by an
additional element $\infty$, 
and $min$ as addition with $\infty$ as neutral element and real addition $+$ as multiplication with
$0$ as neutral element.
The tropical semiring $\cT_{max}$,
also sometimes called {\em arctic semiring}, 
is defined analogously, 
where $\infty$ is replaced by $-\infty$
and $min$ by $max$.
The choice of the 
commutative semiring $\cS$ makes quite a difference as illustrated by the following:
\begin{exs}
\label{ex:1a}
In the next examples the word functions take values in the ring $\Z$ with addition and multiplication,
or in the subsemiring of $\cT_{max}$ generated by $\Z$.
A {\em block of $1$'s in a word $w \in \{0,1\}^{\star}$} is a maximal 
set of consecutive positions $i \in [\ell(w)]$ in the
word $w$ with $P_1(i)$.
\begin{renumerate}
\item
The function $b_1(w)$ counts the number of blocks of $1$'s in $w$.
$b_1(w)$  can be written as
$$
b_1(w)= \sum_{B \subseteq [\ell(w)] : B \mbox{ is a block of 1's}} 1
$$
which is in $\MSOLEVAL_{\Z}$.
Alternatively, it can be written as
\begin{gather}
b_1(w)= \sum_{v \in [\ell(w)] : First-in-Block(v)} 1,
\label{eq:block}
\end{gather}
where $First-in-Block(v)$ is the formula in $\MSOL$ which says that $v$ is a first position in a block of $1$'s.
Equation (\ref{eq:block}) can be expressed 
in $\MSOLEVAL_{\Z}$ and also in both
$\MSOLEVAL_{\cT_{min}}$ and 
$\MSOLEVAL_{\cT_{max}}$.
\item
Let $mb_1^{max}(w)$ be the function which assigns to the word $w$ 
the maximum of the  sizes of 
blocks of $1$'s,
and
$mb_1^{min}(w)$ be the function which assigns to the word $w$ 
the minimum of the  sizes of 
blocks of $1$'s.
One can show, see Remark \ref{rem:hankel}, that
$mb_1^{max}$ 
and
$mb_1^{min}$ 
are not
definable over the ring $\Z$.
However, they are definable over $\cT_{max}$, respectively over $\cT_{min}$, by writing
$$
mb_1^{max} = \max_{B: B \mbox{ is a block of 1's }} \sum_{v: v \in B} 1
$$
and
$$
mb_1^{min} = \min_{B: B \mbox{ is a block of 1's }} \sum_{v: v \in B} 1
$$
\item
The function $b_1(w)^2$ is definable in $\MSOLEVAL_{\Z}$  because $\MSOLEVAL_{\Z}$ is closed under the usual product,
\ifappendix
cf. Appendix \ref{app:inductive}.
\else
cf. Proposition \ref{pro:products-1}.
\fi 
However, it is not definable over either of the two tropical semirings.
To see this one notes that polynomials in a tropical semiring are piecewise linear.
\end{renumerate}
\end{exs}

\begin{remark}
\label{rem:hankel}
Let $f$ be a word function which takes values in a field $\cF$.
The Hankel matrix $\cH(f)$ is the infinite matrix where rows and columns are labeled by
words $u,v$ and the entry $\cH(f)_{u,v} = f(u \circ v)$.
It is shown in \cite{ar:GodlinKotekMakowsky08} that for word functions $f$ in $\MSOLEVAL_{\cF}$ the Hankel matrix $\cH(f)$ has finite rank.
To show non-definability of $f$ it suffices to show that $\cH(f)$ has infinite rank over a field $\cF$ extending $\Z$.
\end{remark}

\ifappendix
\else
\subsection{Formal definition of $\MSOLEVAL$}
\label{appn:inductive}
\label{se:msol}
Let $\cS$ be a commutative semiring, which contains the semiring
of natural numbers $\N$.
We first define $\MSOL$-polynomials, which are multivariate polynomials.
The functions in $\MSOLEVAL$ are obtained from $\MSOL$-polynomials by substituting
values from $\cS$ for the indeterminates.

$\MSOL$-polynomials have a fixed finite set of variables (indeterminates,
if we distinguish them from the variables of $\SOL$), $\mathbf{X}$.
We denote by $\cardm{M}{v}{\varphi}$ the number of elements $v$ 
in the universe that satisfy $\varphi$. 
We 
assume $\tau$ contains a relation symbol $\mathbf{R}_\leq$ 
which is always interpreted as a linear ordering of the universe. 

Let $\mathfrak{M}$ be a $\tau$-structure.
We first define the
{\em $\MSOL(\tau)$-monomials}
inductively.
\begin{definition}[$\MSOL$-monomials]
\ 
\label{def:monomials}
\begin{renumerate}
\item
Let $\phi(v)$ be a formula in $\MSOL(\tau)$,
where $v$ is a
first order variable.
Let $r\in \mathbf{X}\cup\left(\cS-\{0\}\right)$ be either an indeterminate or an integer. 
Then
\[
  r^{\cardm{M}{v}{\phi}}
\]
is a
standard $\MSOL(\tau)$-monomial (whose value depends on $\cardm{M}{v}{\phi}$.
\item
Finite products of $\MSOL(\tau)$-monomials are
$\MSOL(\tau)$-monomials. 
\end{renumerate}
Even if $r$ is an integer, and $r^{\cardm{M}{v}{\phi}}$ does not depend on $\mathfrak{M}$,
the monomial stands as it is, and is not evaluated.
\end{definition}

Note the degree of a monomial is polynomially
bounded by the cardinality of~$\mathfrak{M}$.

\begin{definition}[$\MSOL$-polynomials]
\label{def:polynomials}
The polynomials definable in $\MSOL(\tau)$ are defined inductively: 
\begin{renumerate}
\item
$\MSOL(\tau)$-monomials are
$\MSOL(\tau)$-polynomials. 
\item
Let $\phi$ be a $\tau \cup \{\bar{\mathbf{R}}\}$-formula in $\MSOL$
where 
$\bar{\mathbf{R}} = (\mathbf{R}_1, \ldots , \mathbf{R}_m)$ 
is a finite sequence of {\em unary} relation 
symbols not in $\tau$.
Let $t$ be a
$\MSOL(\tau\cup\{\bar{\mathbf{R}}\})$-polynomial.
Then
$$
\sum_{\bar{R}: \langle \mathfrak{M},\bar{R} \rangle \models \phi(\bar{R})} 
t
$$
is a
$\MSOL(\tau)$-polynomial.
\end{renumerate}
\end{definition}
For simplicity we refer to 
$\MSOL(\tau)$-polynomials as  $\MSOL$-polynomials
when $\tau$ is clear from the context.

We shall use the following properties of $\MSOL$-polynomials.
The proofs  can be found in \cite{ar:KotekMakowskyZilber11}.
\begin{lemma}~
\label{le:constants}
\begin{renumerate}
 \item  Every indeterminate $x\in\mathbf{X}$ can be written as an $\MSOL$-monomial.
 \item Every integer $c$ can be written as an
$\MSOL$-monomial.
 \end{renumerate}
\end{lemma}
\begin{proposition}
\label{pro:products-1}
The pointwise product of two
$\MSOL$-polynomials
is again an
$\MSOL$-polynomial.
\end{proposition}

\fi 
\section{$\MSOLEVAL_{\cS}$ and Weighted Automata}
\label{se:aut}

Let $\cS$ be a commutative semiring and $\Sigma$ a finite alphabet.
A weighted automaton $A$ of size $r$ over $\cS$ is given by:
\begin{renumerate}
\item
Two vectors $\alpha, \gamma \in \cS^r$, and
\item
for each $\sigma \in \Sigma$ a matrix $\mu_{\sigma} \in \cS^{r \times r}$.
\end{renumerate}
For a matrix or vector $M$ we 
denote by $M^T$ the transpose of $M$.

For a word $w = \sigma_1 \sigma_2 \ldots \sigma_{\ell(w)}$
the automaton $A$ defines the function 
\begin{gather}
f_A(w) = \alpha \cdot \mu_{\sigma_1} \cdot \ldots \cdot \mu_{\sigma_{\ell(w)}} \cdot \gamma^{T}.
\notag
\end{gather}
A word function $f: \Sigma^{\star} \rightarrow \cS$ 
is recognized by an automaton $A$ if 
$f=f_A$.
$f$ is recognizable if there exists a weighted automaton $A$ which recognizes it.

\begin{theorem}
\label{th:2nd-a}
\label{th:main-1}
Let $f$ be a word function with values in a commutative semiring $\mathcal{S}$. 
Then 
$f \in \MSOLEVAL_{\mathcal{S}}$ iff 
f is recognized by some weigthed automaton $A$ over $\mathcal{S}$.
\end{theorem}

In this section we prove Theorem \ref{th:main-1} using model theoretic tools, without going through weighted logic.
We need a few definitions.

The {\em quantifier rank $qr(f)$} of a word function $f$ in $\MSOLEVAL_{\cS}$ is defined as the maximal
quantifier rank of the formulas which appear in the definition of $f$. It somehow measures the complexity
of $f$, but we do not need the technical details in this paper. 
Quantifier ranks of formulas in $\MSOL$ are defines as usual, cf. \cite{bk:EF95}.

We denote by  
$\cS^{\Sigma^{\star}}$
the set of word functions $\Sigma^{\star} \rightarrow \cS$.
A {\em semimodule $\cM$} is a subset of
$\cS^{\Sigma^{\star}}$
closed under point-wise addition of word functions in $\cM$, and point-wise multiplication
with elements of $\cS$. 
Note that
$\cS^{\Sigma^{\star}}$ itself is a semimodule.

{\em $M \subseteq \cS^{\Sigma^{\star}}$ is finitely generated} if there is a finite set $F \subseteq \cS^{\Sigma^{\star}}$ 
such that each $f\in M$ 
can be written as a (semiring) linear combination of elements in $F$.
Let $w$ be a word and $f$ a word function. Then we denote by $w^{-1}f$ the word function $g$
defined by
$$
g(u) = (w^{-1}f)(u) = f(w \circ u)
$$
$M$ is {\em stable} if  for all words $w \in \Sigma^{\star}$ and for all $f \in \cM$ the  word function $w^{-1}f$ is also in $M$.

\subsection{ Word functions in $\MSOLEVAL_{\cS}$ are recognizable}
To prove the ``only if'' direction of Theorem \ref{th:main-1}  we use  the following two theorems.

For a commutative semiring $\cS$ and a sequence of indeterminates $\bar{X}=(X_1, \ldots , X_t)$
we denote by $\cS[\bar{X}]$ the commutative semiring of polynomials with indeterminates $\bar{X}$
and coefficients in $\cS$.
\ifappendix
The first theorem is from \cite{ar:MakowskyTARSKI}, and a proof for words is given in Appendix \ref{app:fv}.
\else
The first theorem is from \cite{ar:MakowskyTARSKI}.
\fi 

\begin{theorem}[Bilinear Decomposition Theorem for Word Functions]
\label{th:bdt}
\ \\
Let $\cS$ be a commutative semiring.
Let $f \in \MSOLEVAL_{\cS}$ be a word function $\Sigma^+ \rightarrow \cS$ of quantifier rank $qr(f)$.
There are:
\begin{renumerate}
\item
a function $\beta: \N \rightarrow \N$,
\item
a finite vector $F= (g_1, \ldots , g_{\beta(qr(f))})$ of functions in $\MSOLEVAL_{\cS}$ of length $\beta(qr(f))$,
with $f = g_i$ for some $i \leq \beta(qr(f))$,
\item
and for each $g_i \in F$, a matrix $M^{(i)} \in \cS^{\beta(qr(f)) \times \beta(qr(f))}$
\end{renumerate}
such that
$$g_i(u \circ v) = F(u) \cdot M^{(i)} F(v)^T.$$
\end{theorem}

The other theorem was first proved by G.~Jacob, \cite{phd:Jacob,bk:BerstelReutenauer}.
\begin{theorem}[G. Jacob 1975]
\label{Jacob}
Let $f$ be a word function $f:\Sigma^{\star} \rightarrow \mathcal{S}$. 
Then $f$ is recognizable by a weighted automaton over $\mathcal{S}$ 
iff there exists a finitely generated stable semimodule 
$\mathcal{M} \subseteq \cS^{\Sigma^{\star}}$
which contains $f$.
\end{theorem}

In order to prove the ``only if'' direction of Theorem \ref{th:main-1} we reformulate it. 

\begin{theorem}[Stable Semimodule Theorem]
\label{th:stable}
Let $\cS$ be a commutative semiring and let
$f \in \MSOLEVAL_{\cS}$ be a word function of quantifier rank $qr(f)$.
\\
There are:
\begin{renumerate}
\item
a function $\beta: \N \rightarrow \N$,
\item
a finite vector $F= (g_1, \ldots , g_{\beta(qr(f))})$ of functions in $\MSOLEVAL_{\cS}$ of length $\beta(qr(f))$,
with $f = g_i$ for some $i \leq \beta(qr(f))$,
\end{renumerate}
such that the semimodule $\cM[F]$ generated by $F$ is stable.
\end{theorem}
\begin{proof}
We take $F$ 
and the matrices $M^{(i)}$
\ifappendix
from Theorem \ref{th:bdt} stated in the introduction and proved in Appendix \ref{app:fv}.
\else
from Theorem \ref{th:bdt} stated in the introduction.
\fi 

We have to show that for every fixed word $w$ and $f \in \cM[F]$ the function $w^{-1}f \in \cM[F]$.
As $f \in \cM[F]$  there is a vector 
$A = (a_1, \ldots, a_{\beta(qr(f))}) \in \cS^{\beta(qr(f))}$ 
such that 
\begin{gather}
f(w) = A \cdot F^{T}(w) \notag
\end{gather}
for every fixed word $w$.
Here $F(w)$ is shorthand for $(g_1(w), \ldots, g_{\beta(qr(f))}(w))$.

Let $u$ be a word. We compute $(w^{-1}f)(u)$.
\begin{gather}
(w^{-1}f)(u) = f(w \circ u) = A \cdot F^{T}(w \circ u) = \notag \\ 
\sum_{i=1}^{\beta(qr(f))} a_i g_i(w \circ u) = 
\sum_{i=1}^{\beta(qr(f))} a_i F(w) M^{(i)} F^{T}(u) \notag
\end{gather}
We put $B_i = a_i F(w) M^{(i)}$ and
observe that $B_i \in \cS^{\beta(qr(f))}$.
If we take $B = \sum_i^{\beta(qr(f))} B_i$ we get that $(w^{-1}f)(u) = B \cdot F^T(u)$, hence
$w^{-1}f \in \cM[F]$.
\end{proof}

\subsection{Recognizable  word functions are definable in $\MSOLEVAL_{\cS}$}
For the ``if'' direction we proceed as follows:
\begin{proof}
Let $A$ be a weighted automaton 
of size $r$ over $\mathcal{S}$ for words in $\Sigma^{\star}$.
For a word $w$ with $\ell(w)=n$, given as a function $w: [n] \rightarrow \Sigma$,
the automaton $A$ defines the function 
\begin{gather}
\label{eq:run}
f_A(w) = \alpha \cdot \mu_{w(1)} \cdot \ldots \cdot \mu_{w(n)} \cdot \gamma^{T}.
\end{gather}
We
have to show that
$f_A \in \MSOLEVAL_{\cS}$.

To unify notation we define
\begin{gather}
M_{i,j}^{a}
= (\mu_{a})_{i,j}.
\notag
\end{gather}
Equation (\ref{eq:run})  is a product of $n$ matrices and two vectors.

Let $P$ be the product of these matrices,
\begin{gather}
P = \prod_{k=1}^n \mu_{w(k)}.
\notag
\end{gather}

Using 
matrix algebra we get for the entry $P_{a,b}$ of $P$:
\begin{gather}
P_{a,b} =
\sum_{i_{n-1}=1}^r 
\left(
\sum_{i_{n-2}=1}^r 
\left( \ldots \left(
\sum_{i_{1}=1}^r 
M_{a,i_1}^{w(1)} 
\cdot 
M_{i_1,i_2}^{w(2)} 
\right)
M_{i_2,i_3}^{w(3)} 
\right) \ldots \right)
M_{i_{n-1},b}^{w(n)} 
\notag \\
=
\sum_{i_1, \ldots i_{n-1} \leq r}
\left(
M_{a,i_1}^{w(1)} 
\cdot
M_{i_1,i_2}^{w(2)} 
\cdot
\ldots
\cdot
M_{i_{n-1},b}^{w(n)} 
\right)
\notag 
\end{gather}

Let $\pi: [n-1] \rightarrow [r]$ be the function 
with $\pi(k) =i_k$. We rewrite $P_{a,b}$ as:
\begin{gather}
\label{eq:Pab}
P_{a,b}  = 
\sum_{\pi: [n-1] \rightarrow [r]} \left( 
M_{a,\pi(1)}^{w(1)}  \cdot
M_{\pi(1), \pi(2)}^{w(2)}  \cdot \ldots
M_{\pi(n-1), b}^{w(n)}  \right)
\end{gather}
Next we compute the $b$ coordinate of the vector $\alpha \cdot P$:
\begin{gather}
(\alpha \cdot P)_{b} = 
\sum_{i=1}^r \alpha_i \cdot P_{i,b}
\notag
\end{gather}
Therefore 
\begin{gather}
f_A(w) = \alpha \cdot P \cdot \gamma =
\sum_{b=1}^r (\alpha \cdot P)_b \cdot \gamma_b \notag \\
=
\sum_{b=1}^r 
\left(
\sum_{a=1}^r \alpha_a \cdot P_{a,b}
\right) \cdot \gamma_b =
\sum_{a,b \leq r} \alpha_a \cdot P_{a,b} \cdot \gamma_b
\notag 
\end{gather}
and by using Equation (\ref{eq:Pab}) for $P_{a,b}$ we get:
\begin{gather}
\sum_{a,b \leq r} \alpha_a \cdot 
\left(
\sum_{\pi: [n-1] \rightarrow [r]} \left( 
M_{a,\pi(1)}^{w(1)}  \cdot
M_{\pi(1), \pi(2)}^{w(2)}  \cdot \ldots
M_{\pi(n-1), b}^{w(n)}  \right)
\right)
\cdot \gamma_b
\notag
\end{gather}

Now let $\pi': [n] \cup \{0\} \rightarrow [r]$ be the function for which
$\pi'(0) =a, \pi'(n)=b$ and $\pi'(k) =\pi(k)= i_k$ for $1 \leq  k \leq n-1$.
Then we get

\begin{gather}
\label{eq:run1}
f_A(w) = \notag \\
\sum_{\pi': [n] \cup \{0\} \rightarrow [r]}
\alpha_{\pi'(0)} \cdot 
\left[
M_{\pi'(0),\pi'(1)}^{w(1)} \cdot
\ldots
\cdot
M_{\pi'(n-1),\pi'(n)}^{w(n)} 
\right] 
\cdot
\gamma_{\pi'(n)}= \notag \\
\label{eq:run2}
\sum_{\pi': [n] \cup \{0\} \rightarrow [r]}
\alpha_{\pi'(0)} \cdot
\left(
\prod_{k \in [n]} M_{\pi'(k-1), \pi'(k)}^{ w(k)}
\right) \cdot
\gamma_{\pi'(n)} 
\end{gather}

To convert Equation (\ref{eq:run2}) into an expression in $\MSOLEVAL_{\cS}$ we use a few lemmas:

First, let $S$ be any set and $\pi: S \rightarrow [r]$ be any function.
$\pi$ induces a partition of $S$ into sets 
$U_1^{\pi}, \ldots , U_r^{\pi}$ by
$U_i^{\pi} =\{ s \in S : \pi(s)=i\}$.
Conversely, every partition
$\mathcal{U}= (U_1, \ldots , U_r)$ of $S$ induces a function $\pi_{\mathcal{U}}$ by setting 
$\pi_{\mathcal{U}}(s)= i$ for $s \in U_i$.
To pass between functions $\pi$ with finite range $[r]$ and partitions into $r$-sets we use the following lemma:

\begin{lemma}
\label{le:a}
Let $E(\pi)$ be any expression depending on $\pi$. 
\begin{gather}
\sum_{\pi:S \rightarrow [r]} E(\pi)=
\sum_{\mathcal{U}} E(\pi_{\mathcal{U}})=
\sum_{U_1, \ldots U_r: Partition(U_1, \ldots , U_r)} E(\pi_{\mathcal{U}})
\notag
\end{gather}
where $\mathcal{U}$ ranges over all partitions of $S$ into $r$ sets $U_i: i \in [r]$.
Clearly, $Partition(U_1, \ldots , U_r)$ can be written in $\MSOL$.
\end{lemma}

Second,
to convert the factors 
$\alpha_{\pi'(0)}$ and
$\gamma_{\pi'(n)}$ we proceed as follows:

\begin{lemma}
\label{le:b}
Let $\alpha_i$ be the unique value of the coordinate of $\alpha$
such that $0 \in U_i$.  
Similarly,
let $\gamma_i$ be the unique value of the coordinate of $\gamma$
such that $n \in U_i$.  
\begin{gather}
\alpha_{\pi'(0)}=
\prod_{i=1}^r
\prod_{0 \in U_i} \alpha_i
\notag
\\
\gamma_{\pi'(n)}=
\prod_{i=1}^r
\prod_{n \in U_i} \gamma_i
\notag
\end{gather}
\end{lemma}
\begin{proof}
First we note that, as $\mathcal{U}$ is the partition induced by $\pi'$,
the restriction of $\pi'$ to $U_i$ is constant for all $i \in [r]$.
Next we note that the product ranging over the empty set gives the value $1$.
\end{proof}

Similarly, to convert the factor 
$\prod_{k \in [n]} M_{\pi'(k-1), \pi'(k)}^{ w(k)}$ use the following lemma:

\begin{lemma}
\label{le:c}
Let $m_{i,j, w(v)}$ be the unique value of the $(i,j)$-entry of the matrix $\mu_{w(v)}$
such that $v \in U_i$ and $v+1 \in U_j$.
\begin{gather}
\prod_{k \in [n]} M_{\pi'(k-1), \pi'(k)}^{ w(k)} =
\prod_{i,j=1}^r
\left(
\prod_{v-1 \in U_i, v \in U_j} m_{i,j, w(v)}
\right)
\notag
\end{gather}
\end{lemma}

Using the fact that
every element which is the interpretation of a term in $\cS$  can be written as an expression in $\MSOLEVAL_{\cS}$,
\ifappendix
Lemma \ref{le:constants} in Appendix \ref{app:inductive},
\else
Lemma \ref{le:constants} in Section \ref{appn:inductive},
\fi 
we can write
$U_i(v)$ instead of $v \in U_i$, and
see that the monomials
of Lemmas 
\ref{le:a}, \ref{le:b} and \ref{le:c}  are indeed in $\MSOLEVAL_{\cS}$.
Now we apply the fact that the pointwise product of two word functions in $\MSOLEVAL_{\cS}$ is again a function in
$\MSOLEVAL_{\cS}$,
\ifappendix
Proposition \ref{pro:products-1} in  Appendix \ref{app:inductive}, 
\else
Proposition \ref{pro:products-1} in  Section \ref{appn:inductive}, 
\fi 

to Lemmas
\ref{le:a}, \ref{le:b} and \ref{le:c} 
and complete the proof of Theorem \ref{th:2nd-a}.
\end{proof}

\section{Weighted $\MSOL$ and $\MSOLEVAL$}
\label{se:comp}

In this section we compare the formalism of weighted $\MSOL$, $\WMSOL$, with our $\MSOLEVAL_{\cS}$
for arbitrary commutative semirings.
In \cite{ar:DrosteGastin05,ar:DrosteGastin2007} and \cite{ar:Bollig-etal2010}
two fragments of weighted $\MSOL$ are discussed.
One is based on {\em unambiguous} formulas (a semantic concept), the other on {\em step formulas} 
based on the Boolean fragment of weighted $\MSOL$ (a syntactic definition).
The two fragments have equal expressive power, as stated in \cite{ar:Bollig-etal2010},
and characterize the functions recognizable by weighted automata. 
We denote both versions by $\RMSOL$.

\subsection{Syntax of $\WMSOL$, the weighted version of $\MSOL$}
The definitions and properties of $\WMSOL$ and its fragments are taken literally from \cite{ar:Bollig-etal2010}.
The syntax of formulas $\phi$ of weighted $\MSOL$, denoted by $\WMSOL$, is given inductively in Backus--Naur form by
\begin{gather}
\phi :: = k \mid P_a(x) \mid \neg P_a(x) \mid x \leq y \mid \neg x \leq y \mid x \in X \mid x \not\in X \notag \\
\mid \phi \vee \psi \mid \phi \wedge \psi 
\mid \exists x.\phi \mid \exists X. \phi
\mid \forall x.\phi \mid \forall X. \phi \notag
\end{gather}
where $k \in \cS$, $a \in \Sigma$.
The set of weighted $\MSOL$-formulas over the field $\cS$ and the alphabet $\Sigma$ is denoted
by $\MSOL(\cS,\Sigma)$.
{\em $\bMSOL$ formulas} and {\em $\bMSOL$-step formulas} are defined below.
$\bMSOL$ is the Boolean fragment of $\WMSOL$, and its name
is justified by Lemma \ref{le:bMSOL}.
$\RMSOL$ is the fragment of $\WMSOL$ where 
universal second order quantification is restricted
to $\bMSOL$ and first order universal quantification is restricted to $\bMSOL$-step formulas.

The syntax of weighted $\bMSOL$ is given by
\begin{gather}
\phi :: = 0 \mid 1 \mid 
P_a(x) \mid  x \leq y \mid x \in X 
\mid  \neg \phi \mid \phi \wedge \psi 
\mid \forall x.\phi \mid \forall X. \phi \notag
\end{gather}
where $a \in \Sigma$.

The set of weighted $\MSOL$-formulas over the commutative semiring $\cS$ and the alphabet $\Sigma$ is denoted
by $\WMSOL(\cS,\Sigma)$.

Instead of defining step-formulas as in 
\cite{ar:Bollig-etal2010}
we use Lemma 3 from
\cite{ar:Bollig-etal2010} as our definition.

A $\bMSOL$-step formula $\psi$ is a formula of the form
\begin{gather}
\label{def-step}
\psi = 
\bigvee_{i \in I} ( \phi_i \wedge k_i) 
\end{gather}
where $I$ is a finite set, $\phi_i \in \bMSOL$ and $k_i \in \cS$.

\subsection{Semantics of $\WMSOL$, and translation of $\RMSOL$ into $\MSOLEVAL_{\cS}$}
Next we define the semantics of $\WMSOL$ and, where it is straightforward,
simultaneously also its 
translations  
into $\MSOLEVAL_{\cS}$. 

The evaluations of weighted formulas $\phi \in \WMSOL(\cS,\Sigma)$ on a word $w$ are denoted by
$WE(\phi,w,\sigma)$, where $\sigma$ is an assignment of the variables of $\phi$
to positions, respectively sets of positions, in $w$.

We denote the evaluation of term $t$ of $\MSOLEVAL_{\cS}$ for a word $w$ and an assignment for the free
variables $\sigma$ by $E(t,w,\sigma)$.
$\tw(\phi)$ stands for the truth value of $\phi$
(subject to an assignment for the free variables), i.e., 
$E(\tw(\phi),w,\sigma) =0 \in \cS$ for false and
$E(\tw(\phi),w,\sigma) =1 \in \cS$ for true.
The term $\tw(\phi)$ is used as an abbreviation for
$$
\tw(\phi)= \sum_{U: U=A \wedge \phi} 1  
$$
where $U=A$ stands for $\forall x (U(x) \leftrightarrow x =x)$ and $U$ does not occur freely in $\phi$.
Indeed, we have
$$
E(\tw(\phi),w,\sigma) = \begin{cases} 1 & (w, \sigma) \models \phi \\ 0&  \mbox{ else } \end{cases}
$$

We denote by $TRUE(x)$ the formula $x=x$ with free first order variable $x$.
Similarly, $TRUE(X)$ denotes the formula $\exists y \in X \vee \neg \exists y \in X$ with free set variable $X$. 

The evaluations of formulas $\phi \in \WMSOL$ and their translations are now defined inductively.
\begin{renumerate}
\item
For $k \in \cS$ we have  $tr(k) =k$ and
$WE(k,w,\sigma)) = E(tr(k) ,w,\sigma)) = k$.
\item
For atomic formulas $\theta$ we have 
$tr(\theta) = \tw(\theta)$ 
and 
$$WE(\theta,w,\sigma) = 
E(tr(\theta),w, \sigma) =
E(\tw(\theta),w, \sigma)$$
\item
For negated atomic formulas we have
$$tr(\neg \theta) = 1 - tr(\theta) = 1 - \tw(\theta)$$ 
and
$$WE(\neg \theta,w, \sigma)= 
1 - E(\tw(\theta),w, \sigma).$$
\item
$tr(\phi_1 \vee \phi_2) = tr(\phi_1) + tr(\phi_2)$ and
$$WE(\phi_1 \vee \phi_2, w, \sigma)= E(tr(\phi_1) + tr(\phi_2), w, \sigma) =
E(tr(\phi_1), w , \sigma) + E(tr(\phi_2), w, \sigma).$$
\item
$tr(\exists x. \phi) = \sum_{x: TRUE(x)} tr(\phi)$ and
$$WE(\exists x. \phi, w , \sigma)= 
E(\sum_{x: TRUE(x)} tr(\phi, w, \sigma))=
\sum_{x: TRUE(x)} E(tr(\phi, w, \sigma)).$$
\item
$tr(\exists X. \phi) = \sum_{X:  TRUE(X) } tr(\phi)$ and
$$WE(\exists X. \phi, w , \sigma)= 
E(\sum_{X:  TRUE(X) } tr(\phi, w, \sigma))=
\sum_{X: TRUE(X) } E(tr(\phi, w, \sigma)).$$
\item
$tr(\phi_1 \wedge \phi_2) = tr(\phi_1) \cdot tr(\phi_2)$ and
$$WE(\phi_1 \wedge \phi_2, w, \sigma)=
E(tr(\phi_1) \cdot tr(\phi_2), w, \sigma) =
E(tr(\phi_1), w , \sigma) \cdot E(tr(\phi_2), w \sigma).$$
\setcounter{ourown}{\value{enumi}}
\end{renumerate}
So far the definition of $WE$ was given using the evaluation function $E$
and the translation was straightforward.
Problems arise with the universal quantifiers.

The unrestricted definition of $WE$ for $\WMSOL$
given below gives us functions which are not recognizable by weighted automata,
and the straightforward translation defined below gives us expressions which are not in $\MSOLEVAL_{\cS}$:
\begin{renumerate}
\addtocounter{enumi}{\value{ourown}}
\item
$tr(\forall x. \phi) = \prod_{x: TRUE(x)} tr(\phi)$ and
$$WE(\forall x. \phi, w , \sigma)= 
E(\prod_{x: TRUE(x)} tr(\phi, w, \sigma))=
\prod_{x: TRUE(x)} E(tr(\phi, w, \sigma)).$$

The formula $\phi_{sq}= \forall x. \forall y. 2$
gives the function $2^{\ell(w)^2}$
and is not a $\bMSOL$-step formula.
The straightforward translation $tr$ gives
the term 
\begin{gather}
\prod_{x: TRUE(x)} \left( 
\prod_{y: TRUE(y)}  2 \right) = \prod_{ (x,y) : TRUE(x,y)} 2 ,
\notag
\end{gather} 
which
is a product over the tuples of a binary relation, hence not in $\MSOLEVAL_{\cS}$.
\item
$tr(\forall X. \phi) = \prod_{X: TRUE(X) } tr(\phi)$ and
$$WE(\forall X. \phi, w , \sigma)= 
E(\prod_{X:  TRUE(X) } tr(\phi, w, \sigma))=
\prod_{X: TRUE(X) } E(tr(\phi, w, \sigma)).$$

Here the translation gives a  product $\prod_{X:TRUE(X)}$ ranging over subsets, which is not an expression in $\MSOLEVAL_{\cS}$.
\end{renumerate}

In $\RMSOL$,
universal second order quantification is restricted  to formulas of
$\bMSOL$, and first order universal quantification is restricted to $\bMSOL$-step formulas.

In \cite[page 590]{ar:Bollig-etal2010}, after Figure 1, the following is stated:

\begin{lemma}
\label{le:bMSOL}
The evaluation $WE$ of a $\bMSOL$-formula $\phi$ assumes values in $\{0,1\}$ and coincides with
the standard semantics of $\phi$ as an unweighted $\MSOL$-formula.
\end{lemma}

Because the translation of universal quantifiers using
$tr$ leads outside of $\MSOLEVAL_{\cS}$,
we 
define a proper translation 
$tr': \RMSOL \rightarrow \MSOLEVAL_{\cS}$.

Using Lemma \ref{le:bMSOL} we set $tr'(\phi)=\tw(\phi)$, for $\phi$ a $\bMSOL$-formula.

For universal first order quantification of $\bMSOL$-step formulas  
\begin{gather}
\psi = 
\bigvee_{i \in I} ( \phi_i \wedge k_i) 
\end{gather}
we compute $WE(\forall x.\psi, w , \sigma)$ and 
$E(tr(\forall x.\psi), w, \sigma)$ 
as follows, leaving the steps for
the translation  of $tr(\forall x.\psi)$ to the reader.

\begin{gather}
WE((\forall x. \psi), w, \sigma ) = 
E( tr(\forall x. \psi), w, \sigma ) =  \notag \\
E(tr( \forall x. \bigvee_{i \in I} ( \phi_i \wedge k_i) ), w , \sigma ) = \notag \\
\prod_{x: TRUE(x)} E(tr(\bigvee_{i \in I} ( \phi_i \wedge k_i))) , w , \sigma )= \notag \\
\prod_{x: TRUE(x)} (\sum_{i \in I} ( E(tr'(\phi_i)) \cdot k_i), w , \sigma ) ) = \notag \\
\prod_{x: TRUE(x)} (\sum_{i \in I} ( E(\tw(\phi_i), w , \sigma ) \cdot k_i)) ) \notag 
\end{gather}
Clearly, the formula of the last line,
$\prod_{x: TRUE(x)} (\sum_{i \in I} ( \tw(\phi_i)) \cdot k_i) )$
is an expression in $\MSOLEVAL_{\cS}$.

For universal second order quantification of $\bMSOL$-formulas   $\psi$
we use Lemma \ref{le:bMSOL} and get

\begin{gather}
WE(\forall X. \psi, w , \sigma ) =  E(tr'( \forall X \psi), w , \sigma ) = E(\tw( \forall X \psi), w , \sigma ) 
\notag
\end{gather}
Clearly, the expression $\tw( \forall X \psi)$
is an expression in $\MSOLEVAL_{\cS}$.
Thus we have proved:
\begin{theorem}
\label{th:transl}
Let $\cS$  be a commutative semiring.
For every expression $\phi \in \RMSOL$ there is an expression $tr'(\phi) \in \MSOLEVAL_{\cS}$
such that
$WE(\phi,w, \sigma) = E(tr'(\phi), w, \sigma)$,
i.e., $\phi$ and $tr'(\phi)$ define the same word function.
\end{theorem}

\subsection{Translation from $\MSOLEVAL_{\cS}$ to $\RMSOL$}

It follows from our Theorem \ref{th:main-1} and the characterization in \cite{ar:DrosteGastin2007}
of recognizable word functions
as the functions definable in $\RMSOL$,
that the converse is also true.
We now give a direct proof of the converse without using weighted automata.
\begin{theorem}
\label{th:transl-1}
Let $\cS$  be a commutative semiring.
For every expression 
$t \in \MSOLEVAL_{\cS}$ there is a formula
$\phi_t \in \RMSOL$ 
such that
$WE(\phi_t,w, \sigma) = E(t, w, \sigma)$,
i.e., $\phi_t$ and $t$ define the same word function.
\end{theorem}
\begin{proof}
\begin{renumerate}
\item
Let $t = \prod_{x: \phi(x)} \alpha$ be a $\MSOLEVAL_{\cS}$- monomial.
We note that 
$$
\alpha \cdot \tw(\phi)+ \tw(\neg \phi) =
\begin{cases}
\alpha &  \mbox{ if } \phi \mbox{ is true }\\ 
1 &  \mbox{else}
\end{cases}
$$
Furthermore, by Lemma \ref{le:bMSOL} $\phi \in \bMSOL$.
So we  put 
$$\phi_t = \forall x. ( (\phi(x) \wedge \alpha) \vee \neg \phi(x))$$
\item
Let $t_1 = \sum_{U: \phi(U)} t$ and let $\phi_t$ be the translation of $t$.
Then
$$
\phi_{t_1} =\exists U. (\phi_t \wedge \phi(U))
$$
\end{renumerate}
\end{proof}
\section{Conclusions}
\label{se:conclu}

We have given two proofs that $\RMSOL$ and $\MSOLEVAL$ with values in $\cS$ 
have the same expressive power over words. 
One proof uses model theoretic tools to show directly that $\MSOLEVAL$ captures the functions recognizable by weighted automata.
The other proof shows how to translate the formalisms from one into the other.
Adapting the translation proof, it should be possible to extend the result to tree functions as well, cf. \cite{ar:DrosteVogler2006}.

Although in this paper we dealt only with word functions, 
our formalism $\MSOLEVAL$, introduced first fifteen years ago, was originally designed to
deal with definability of graph parameters and graph polynomials,
\cite{ar:CourcelleMakowskyRoticsDAM,ar:MakowskyTARSKI,ar:MakowskyZoo,ar:KotekMakowskyZilber11}. 
It has been useful, since,  in many applications in algorithmic and structural graph theory and descriptive complexity.
Its use in characterizing word functions recognizable by weighted automata is new.
$\MSOLEVAL$ can be seen as an analogue of the {\em Skolem elementary functions} aka {\em lower elementary functions}, \cite{ar:Skolem1962,ar:Volkov2010},
adapted to the framework of {\em meta-finite model theory} as defined in \cite{ar:GraedelGurevich}.

The formalism $\WMSOL$ of weighted logic was first invented in 2005 in \cite{ar:DrosteGastin05} and since then used
to characterize word and tree functions recognizable by weighted automata, \cite{ar:DrosteVogler2006}.
These characterizations need some syntactic restrictions which lead to the formalisms of $\RMSOL$.
No such syntactic restrictions are need for the characterization of recognizable word functions using $\MSOLEVAL$.
The weighted logic $\WMSOL$ can also be defined for general relational structures.
However, it is not immediate which syntactic restrictions are needed, if at all,  to obtain
algorithmic applications similar to the ones obtained using $\MSOLEVAL$, 
cf. \cite{ar:CourcelleMakowskyRoticsDAM,ar:CourcelleMakowskyRoticsTOCS,ar:Makowsky2005}.
\paragraph*{\footnotesize Acknowledgements}
{
\footnotesize
We are thankful to 
Jacques Sakarovitch,
G\'eraud S\'enizergues  and
Amir Shpilka 
for useful guidence on the
subject of weighted automata.
We are thankful to 
Manfred Droste, 
Tomer Kotek and Elena Ravve 
and several anonymous readers,
for their valuable comments.
}
\nocite{*}
\bibliographystyle{eptcs}
\bibliography{gandalf-final}
\newpage
\ifappendix
\appendix
\input{ams-inductive}
\input{appendix-bdt-words}
\else
\fi 
\end{document}

\ifskip 
\else
\section{Introduction}

The optional arguments of {\tt $\backslash$documentclass$\{$eptcs$\}$} are
\begin{itemize}
\item at most one of
{\tt adraft},
{\tt submission} or
{\tt preliminary},
\item at most one of {\tt publicdomain} or {\tt copyright},
\item and optionally {\tt creativecommons},
  \begin{itemize}
  \item possibly augmented with
    \begin{itemize}
    \item {\tt noderivs}
    \item or {\tt sharealike},
    \end{itemize}
  \item and possibly augmented with {\tt noncommercial}.
  \end{itemize}
\end{itemize}
We use {\tt adraft} rather than {\tt draft} so as not to confuse hyperref.
The style-file option {\tt submission} is for papers that are
submitted to {\tt $\backslash$event}, where the value of the latter is
to be filled in in line 2 of the tex-file. Use {\tt preliminary} only
for papers that are accepted but not yet published. The final version
of your paper that is to be uploaded at the EPTCS website should have
none of these style-file options.

By means of the style-file option
\href{http://creativecommons.org/about/license/}{creativecommons}
authors equip their paper with a Creative Commons license that allows
everyone to copy, distribute, display, and perform their copyrighted
work and derivative works based upon it, but only if they give credit
the way you request. By invoking the additional style-file option {\tt
noderivs} you let others copy, distribute, display, and perform only
verbatim copies of your work, but not derivative works based upon
it. Alternatively, the {\tt sharealike} option allows others to
distribute derivative works only under a license identical to the
license that governs your work. Finally, you can invoke the option
{\tt noncommercial} that let others copy, distribute, display, and
perform your work and derivative works based upon it for
noncommercial purposes only.

Authors' (multiple) affiliations and emails use the commands
{\tt $\backslash$institute} and {\tt $\backslash$email}.
Both are optional.
Authors should moreover supply
{\tt $\backslash$titlerunning} and {\tt $\backslash$authorrunning},
and in case the copyrightholders are not the authors also
{\tt $\backslash$copyrightholders}.
As illustrated above, heuristic solutions may be called for to share
affiliations. Authors may apply their own creativity here.

Exactly 46 lines fit on a page.
The rest is like any normal {\LaTeX} article.
We will spare you the details.
The rest is like any normal {\LaTeX} article.
We will spare you the details.\\
The rest is like any normal {\LaTeX} article.
We will spare you the details.\\
The rest is like any normal {\LaTeX} article.
We will spare you the details.\\
The rest is like any normal {\LaTeX} article.
We will spare you the details.\\
The rest is like any normal {\LaTeX} article.
We will spare you the details.\hfill6\\
The rest is like any normal {\LaTeX} article.
We will spare you the details.\\
The rest is like any normal {\LaTeX} article.
We will spare you the details.\\
The rest is like any normal {\LaTeX} article.
We will spare you the details.\\
The rest is like any normal {\LaTeX} article.
We will spare you the details.\\
The rest is like any normal {\LaTeX} article.
We will spare you the details.\hfill11\\
The rest is like any normal {\LaTeX} article.
We will spare you the details.\\
The rest is like any normal {\LaTeX} article.
We will spare you the details.

Here starts a new paragraph. The rest is like any normal {\LaTeX} article.
We will spare you the details.
The rest is like any normal {\LaTeX} article.
We will spare you the details.\\
The rest is like any normal {\LaTeX} article.
We will spare you the details.\hfill16\\
The rest is like any normal {\LaTeX} article.
We will spare you the details.\\
The rest is like any normal {\LaTeX} article.
We will spare you the details.\\
The rest is like any normal {\LaTeX} article.
We will spare you the details.\\
The rest is like any normal {\LaTeX} article.
We will spare you the details.\\
The rest is like any normal {\LaTeX} article.
We will spare you the details.\hfill21\\
The rest is like any normal {\LaTeX} article.
We will spare you the details.\\
The rest is like any normal {\LaTeX} article.
We will spare you the details.\\
The rest is like any normal {\LaTeX} article.
We will spare you the details.\\
The rest is like any normal {\LaTeX} article.
We will spare you the details.\\
The rest is like any normal {\LaTeX} article.
We will spare you the details.\hfill26\\
The rest is like any normal {\LaTeX} article.
We will spare you the details.\\
The rest is like any normal {\LaTeX} article.
We will spare you the details.\\
The rest is like any normal {\LaTeX} article.
We will spare you the details.\\
The rest is like any normal {\LaTeX} article.
We will spare you the details.\\
The rest is like any normal {\LaTeX} article.
We will spare you the details.\hfill31\\
The rest is like any normal {\LaTeX} article.
We will spare you the details.\\
The rest is like any normal {\LaTeX} article.
We will spare you the details.\\
The rest is like any normal {\LaTeX} article.
We will spare you the details.\\
The rest is like any normal {\LaTeX} article.
We will spare you the details.\\
The rest is like any normal {\LaTeX} article.
We will spare you the details.\hfill36\\
The rest is like any normal {\LaTeX} article.
We will spare you the details.\\
The rest is like any normal {\LaTeX} article.
We will spare you the details.\\
The rest is like any normal {\LaTeX} article.
We will spare you the details.\\
The rest is like any normal {\LaTeX} article.
We will spare you the details.\\
The rest is like any normal {\LaTeX} article.
We will spare you the details.\hfill41\\
The rest is like any normal {\LaTeX} article.
We will spare you the details.\\
The rest is like any normal {\LaTeX} article.
We will spare you the details.\\
The rest is like any normal {\LaTeX} article.
We will spare you the details.\\
The rest is like any normal {\LaTeX} article.
We will spare you the details.\\
The rest is like any normal {\LaTeX} article.
We will spare you the details.\hfill46\\
The rest is like any normal {\LaTeX} article.
We will spare you the details.
The rest is like any normal {\LaTeX} article.
We will spare you the details.
The rest is like any normal {\LaTeX} article.
We will spare you the details.
The rest is like any normal {\LaTeX} article.
We will spare you the details.
The rest is like any normal {\LaTeX} article.
We will spare you the details.
The rest is like any normal {\LaTeX} article.
We will spare you the details.
The rest is like any normal {\LaTeX} article.
We will spare you the details.
The rest is like any normal {\LaTeX} article.
We will spare you the details.
The rest is like any normal {\LaTeX} article.
We will spare you the details.
The rest is like any normal {\LaTeX} article.
We will spare you the details.
The rest is like any normal {\LaTeX} article.
We will spare you the details.
The rest is like any normal {\LaTeX} article.
We will spare you the details.
The rest is like any normal {\LaTeX} article.
We will spare you the details.
The rest is like any normal {\LaTeX} article.
We will spare you the details.
The rest is like any normal {\LaTeX} article.
We will spare you the details.
The rest is like any normal {\LaTeX} article.
We will spare you the details.
The rest is like any normal {\LaTeX} article.
We will spare you the details.

\section{Prefaces}

Volume editors may create prefaces using this very template,
with {\tt $\backslash$title$\{$Preface$\}$} and {\tt $\backslash$author$\{\}$}.

\section{Bibliography}

We request that you use
\href{http://www.cse.unsw.edu.au/~rvg/EPTCS/eptcs.bst}
{\tt $\backslash$bibliographystyle$\{$eptcs$\}$}
\cite{bibliographystylewebpage}. Compared to the original {\LaTeX}
{\tt $\backslash$biblio\-graphystyle$\{$plain$\}$},
it ignores the field {\tt month}, and uses the extra
bibtex fields {\tt eid}, {\tt doi}, {\tt ee} and {\tt url}.
The first is for electronic identifiers (typically the number $n$
indicating the $n^{\rm th}$ paper in an issue) of papers in electronic
journals that do not use page numbers. The other three are to refer,
with life links, to electronic incarnations of the paper.

Almost all publishers use digital object identifiers (DOIs) as a
persistent way to locate electronic publications. Prefixing the DOI of
any paper with {\tt http://dx.doi.org/} yields a URI that resolves to the
current location (URL) of the response page\footnote{Nowadays, papers
  that are published electronically tend
  to have a \emph{response page} that lists the title, authors and
  abstract of the paper, and links to the actual manifestations of
  the paper (e.g.\ as {\tt dvi}- or {\tt pdf}-file). Sometimes
  publishers charge money to access the paper itself, but the response
  page is always freely available.}
of that paper. When the location of the response page changes (for
instance through a merge of publishers), the DOI of the paper remains
the same and (through an update by the publisher) the corresponding
URI will then resolve to the new location. For that reason a reference
ought to contain the DOI of a paper, with a life link to corresponding
URI, rather than a direct reference or link to the current URL of
publisher's response page. This is the r\^ole of the bibtex field {\tt doi}.
DOIs of papers can often be found through
\url{http://www.crossref.org/guestquery};\footnote{For papers that will appear
  in EPTCS and use \href{http://www.cse.unsw.edu.au/~rvg/EPTCS/eptcs.bst}
  {\tt $\backslash$bibliographystyle$\{$eptcs$\}$} there is no need to
  find DOIs on this website, as EPTCS will look them up for you
  automatically upon submission of a first version of your paper;
  these DOIs can then be incorporated in the final version, together
  with the remaining DOIs that need to found at DBLP or publisher's webpages.}
the second method {\it Search on article title}, only using the {\bf
surname} of the first-listed author, works best.  
Other places to find DOIs are DBLP and the response pages for cited
papers (maintained by their publishers).
{\bf EPTCS requires the inclusion of a DOI in each cited paper, when available.}

Often an official publication is only available against payment, but
as a courtesy to readers that do not wish to pay, the authors also
make the paper available free of charge at a repository such as
\url{arXiv.org}. In such a case it is recommended to also refer and
link to the URL of the response page of the paper in such a
repository.  This can be done using the bibtex fields {\tt ee} or {\tt
url}, which are treated as synonyms.  These fields should not be used
to duplicate information that is already provided through the DOI of
the paper.
You can find archival-quality URL's for most recently published papers
in DBLP---they are in the bibtex-field {\tt ee}. In fact, it is often
useful to check your references against DBLP records anyway, or just find
them there in the first place.

When using {\LaTeX} rather than {\tt pdflatex} to typeset your paper, by
default no linebreaking within long URLs is allowed. This leads often
to very ugly output, that moreover is different from the output
generated when using {\tt pdflatex}. This problem is repaired when
invoking \href{http://www.cse.unsw.edu.au/~rvg/EPTCS/breakurl.sty}
{\tt $\backslash$usepackage$\{$breakurl$\}$}: it allows linebreaking
within links and yield the same output as obtained by default with
{\tt pdflatex}. 
When invoking {\tt pdflatex}, the package {\tt breakurl} is ignored.
\fi 

\nocite{*}
\bibliographystyle{eptcs}
\bibliography{generic}
\end{document}